\documentclass[aps, twocolumn, prl,superscriptaddress,footinbib]{revtex4-2}

\usepackage{jefferymath}
\usepackage[colorlinks=true, linkcolor=blue]{hyperref}

\usepackage{algpseudocodex}
\usepackage{algorithm}
\usepackage{CJKutf8}

\let\originalleft\left
\let\originalright\right
\renewcommand{\left}{\mathopen{}\mathclose\bgroup\originalleft}
\renewcommand{\right}{\aftergroup\egroup\originalright}

\newtheorem{theorem}{Theorem}
\newtheorem{lemma}{Lemma}

\begin{document}

\begin{CJK}{UTF8}{gbsn}
\title{Efficient preparation of Dicke states}
\author{Jeffery Yu}
\affiliation{Joint Center for Quantum Information and Computer Science, NIST/University of Maryland, College Park, Maryland 20742, USA}
\affiliation{Joint Quantum Institute, NIST/University of Maryland, College Park, Maryland 20742, USA}
\affiliation{Department of Physics, University of Maryland, College Park, Maryland 20742, USA}
\author{Sean R. Muleady}\thanks{These two authors contributed equally.}
\affiliation{Joint Center for Quantum Information and Computer Science, NIST/University of Maryland, College Park, Maryland 20742, USA}
\affiliation{Joint Quantum Institute, NIST/University of Maryland, College Park, Maryland 20742, USA}
\author{Yu-Xin Wang (王语馨)}\thanks{These two authors contributed equally.}
\affiliation{Joint Center for Quantum Information and Computer Science, NIST/University of Maryland, College Park, Maryland 20742, USA}
\author{Nathan~Schine}
\affiliation{Joint Quantum Institute, NIST/University of Maryland, College Park, Maryland 20742, USA}
\affiliation{Department of Physics, University of Maryland, College Park, Maryland 20742, USA}
\author{Alexey V.\ Gorshkov}
\affiliation{Joint Center for Quantum Information and Computer Science, NIST/University of Maryland, College Park, Maryland 20742, USA}
\affiliation{Joint Quantum Institute, NIST/University of Maryland, College Park, Maryland 20742, USA}
\author{Andrew M.\ Childs}
\affiliation{Joint Center for Quantum Information and Computer Science, NIST/University of Maryland, College Park, Maryland 20742, USA}
\affiliation{Department of Computer Science and Institute for Advanced Computer Studies, University of Maryland, College Park, Maryland 20742, USA}

\begin{abstract}
We present an algorithm utilizing mid-circuit measurement and feedback that prepares Dicke states with polylogarithmically many ancillas and polylogarithmic depth. Our algorithm uses only global mid-circuit projective measurements and adaptively-chosen global rotations.
This improves over prior work that was only efficient for Dicke states of low weight, or was not efficient in both depth and width. Our algorithm can also naturally be implemented in a cavity QED context using polylogarithmic time, zero ancillas, and atom-photon coupling scaling with the square root of the system size.
\end{abstract}

\maketitle
\end{CJK}

Preparing entangled states is a central goal in quantum science. Dicke states are a key class of entangled states~\cite{Dicke_1954}. The weight-$w$ Dicke state for a system of $n$ spin-$\half$ particles is the symmetric equal-weight superposition of all configurations where $w$ particles are spin down and $n-w$ are spin up. Such states are total angular momentum eigenstates, and have featured prominently in the study of quantum optics~\cite{Dicke_1954} and quantum magnetism~\cite{anderson_approximate_1952}. More recently, in the context of quantum information science, Dicke states have emerged as a potential resource for quantum sensing~\cite{sorensen_entanglement_2001,Toth_2012,zhang_quantum_2014, Ouyang_2022} and quantum algorithms~\cite{Childs2002, farhi_quantum_2014, Hadfield_2019, Cook_2020}.
Beyond these applications, it is also of general interest to understand the capabilities and limitations of shallow adaptive quantum circuits, for which preparing such entangled states is a natural milestone.

Despite the theoretical simplicity of Dicke states, their preparation remains an outstanding challenge, particularly considering the limited coherence time and capabilities of current, noisy intermediate-scale quantum (NISQ) devices~\cite{Preskill_2018}. Protocols for preparing Dicke states have been widely studied both in the abstract circuit model and in various experimental setups; see for example Refs.~\cite{Childs2002, 2019deterministic, johnsson_geometric_2020, Wang_2021, Bartschi_2022, aktar_divide-and-conquer_2022-1,buhrman2023laqcc, piroli2024approximating, raveh_dicke_2024,raveh_q-analog_2024, nepomechie_spin-s_2024} and~\cite{kiesel_experimental_2007, wieczorek_experimental_2009,prevedel_experimental_2009,hume_preparation_2009,vanderbruggen_spin-squeezing_2011, bucker_twin-atom_2011,toyoda_generation_2011,chiuri_experimental_2012,noguchi_generation_2012,haas_entangled_2014,lucke_detecting_2014,opatrny_counterdiabatic_2016,zou_beating_2018,davis_painting_2018, carrasco_dicke_2023, chen_-chip_2023,bond_efficient_2023,saleem_achieving_2024}, respectively. Among existing circuits, the vast majority are only efficient for preparing low-weight Dicke states; when the desired weight scales with the system size, with $w = \frac{n}{2}$ being the hardest case, these circuits have depths scaling as $\Omega(n^{1/4})$. We are only aware of one prior work \cite{buhrman2023laqcc} that achieves a circuit depth of $O(\log n)$ for all weights $w$, but it uses $O(n^2 \log n)$ ancilla qubits \footnote{We are also aware of concurrent work by Liu, Childs, and Gottesman that also gives an algorithm for preparing Dicke states in $O(\log n)$ depth, using $O(n \log n)$ ancillas \cite{LCG24}. The approaches are very different: while we apply a simple sequence of collective rotations and collective $J_z$ measurements, the other approach uses sorting networks. Our approach is simpler, likely performs better in practice, and can perform better when the Hamming weight is lower, while the other approach solves a more general symmetrization problem.}.

Here, we present a simple algorithm whose circuit implementation is polylogarithmic in both depth and number of ancillas. Starting from the all spin-up state, we perform the same rotation on each spin, followed by measuring the collective magnetization, which projects onto a Dicke state. If we measure the desired weight $w$, we are done. If we measure some other $w'$, we rotate again by some angle conditioned on $w'$, perform another measurement, and repeat until we obtain $w$.
The main technical contributions of this work are the choice of the rotation angles and the analysis of the expected number of iterations to reach $w$.

The choice of rotation angles is motivated by a geometric phase-space representation of the Husimi-Q distribution as rings on a collective Bloch sphere. In this model, collective rotations correspond to rotating the ring on the sphere, and measurements correspond to projecting the rotated ring onto a Dicke ring, as shown in Figure \ref{fig:geometric-procedure}. We choose the rotation angle so that the resulting ring has maximal overlap with the ring for the desired Dicke state.

\begin{figure*}[t]
\centering
\includegraphics[width=0.9\textwidth]{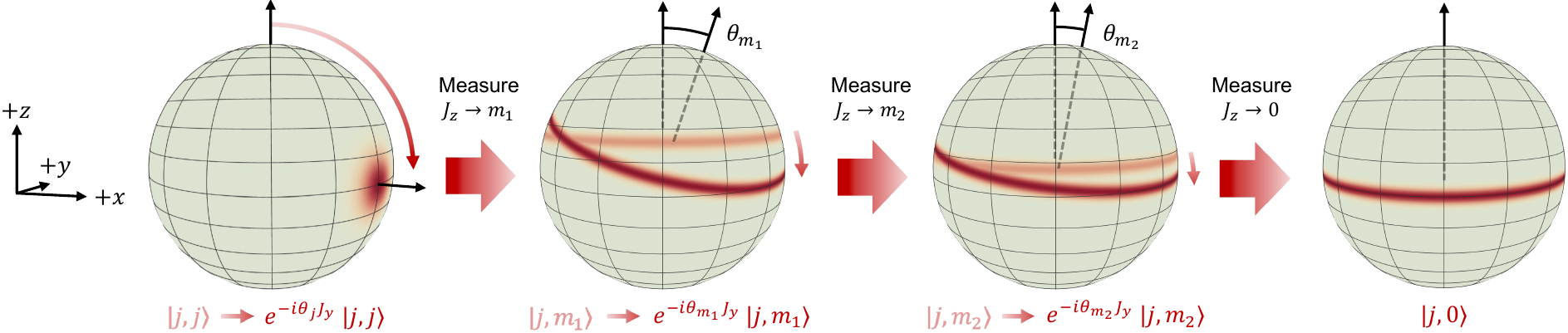}
\caption{A geometric representation of Algorithm~\ref{alg} using the Husimi-Q distribution for the Dicke states, represented as rings on the collective Bloch sphere. For each iteration $i$, the algorithm rotates the current state $|j,m_i\rangle$ by angle $\theta_{m_i}$ about $+y$, so that the corresponding ring is tangent to the ring of the target Dicke state, maximizing their overlap. For a target $m_t = 0$, the corresponding ring lies at the Bloch sphere equator. We project with a $J_z$ measurement, and repeat until we measure the desired state.}
\label{fig:geometric-procedure}
\end{figure*}

A central component of our protocol is the ability to perform mid-circuit collective $J_z$ measurements. Notably, this differs from measuring $\sigma_z$ for each particle individually. In our analysis of circuit complexity, we use the fact that a collective $J_z$ measurement can be implemented in logarithmic depth \cite{piroli2024approximating, zi2024shallow, 2024loghamming}.
On the other hand, going beyond the circuit model, prior work in cavity quantum electrodynamics has explored how such measurements may be heralded by cavity photons or be used in continuous-measurement-and-feedback schemes to generate complex many-body entangled states~\cite{leroux_implementation_2010,schleier-smith_squeezing_2010,christensen_quantum_2014,chen_carving_2015,cox_deterministic_2016,davis_painting_2018,deist_mid-circuit_2022-1,ramette_counter-factual_2024,hartung_quantum-network_2024,grinkemeyer_error-detected_2024}. Using these ideas, we also discuss a constant-time implementation of this measurement in cavity systems. The ability to efficiently prepare Dicke states using only global operations opens the door to harnessing these states for sensing tasks, where they offer a quantum-enhanced precision that can achieve Heisenberg-limited scaling, while offering robustness compared to other entangled resources~\cite{saleem_achieving_2024,lin_covariant_2024}.

\textit{Preliminaries.---}Dicke states have a convenient representation in terms of angular momentum eigenstates, where the $n$ qubits are viewed as spin-$\half$ particles. Let $\mathbf{S}_k$ be the angular momentum operator for the $k$th qubit, and let
\(
\mathbf{J} = \sum_{k=1}^n \mathbf{S}_k
\)
be the total angular momentum operator.
The Dicke states $\ket{j,m}$ are simultaneous eigenstates of $\mathbf{J}^2$ and $J_z$
with quantum numbers $j = \frac{n}{2}$
and $m \in \{-j, -j+1, \dots, j\}$, respectively.
The Dicke state $\ket{j,m}$ is simply the uniform superposition of strings of $n=2j$ bits with Hamming weight $w=j-m$.

We assume the following primitives for our model of computation: (1) prepare $\ket{0}^{\otimes n}$ states on demand, (2) perform collective rotations (uniform single-qubit gates) about the $y$-axis, and (3) perform $J_z$ measurements. Note that both collective rotations, expressed via the unitary $e^{-i\theta J_y}$, and $J_z$ measurements preserve permutation symmetry, leaving the quantum number $j = \frac{n}{2}$ fixed.

Below, we present an experimental setup where a collective $J_z$ measurement can be implemented in $O(1)$ time, independent of $n$.
Even without access to collective measurements, we can implement a $J_z$ measurement in the ordinary circuit model, where such a measurement is equivalent to a projective Hamming weight measurement, with only $\log(n)$ overhead.
One approach is to set $\floor{\log(n)} + 1$ ancillas to be the measurement register and accumulate the Hamming weight into those qubits. Reference~\cite{zi2024shallow} gives an implementation of this in $O(\log^2 n)$ depth with no additional ancillas.

\begin{figure*}
\centering
\includegraphics[width=.97\textwidth]{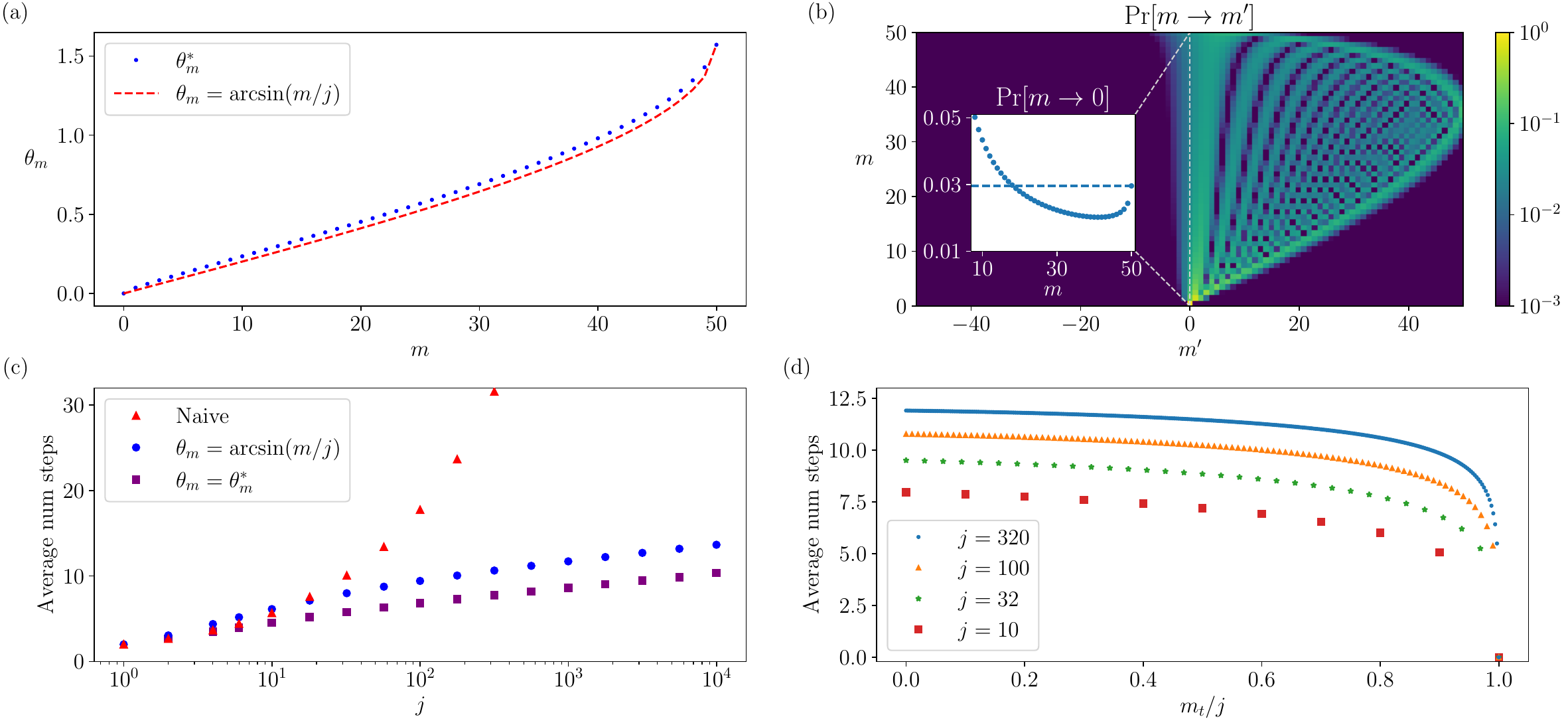}
\caption{(a) For $j=50$ and $m_t=0$, comparison of the numerically computed $\theta_m^*$ and the chosen $\theta_m = \arcsin\bigl(\smash{\frac{m}{j}}\bigr)$. (b) For the same parameters, transition probability matrix $\Pr[m \to m']$ of the base algorithm without resets, which is symmetric about the origin, i.e., $\Pr[-m \to -m'] = \Pr[m \to m']$. Inset shows the transition probabilities for the $m'=0$ slice.
(c) Expected running time for preparing the $m_t=0$ state with various algorithms. Our algorithm exhibits similar logarithmic scaling using the geometric angles (blue circles) as with the numerically optimal angles (purple squares), both of which are exponentially faster than the polynomial running time with the naive approach of resetting at every step (red triangles).
(d) Expected runtime of our algorithm using the geometric angle for any target state. In all cases, the $m_t=0$ state takes the longest.}
\label{fig:main}

\end{figure*}

\textit{Algorithm.---}Our goal is to prepare the Dicke state $\ket{j,m_t}$ for a desired target value of $m_t$, starting from the initial product state $\ket{j,j} = \ket{0}^{\otimes n}$. 
The basic algorithm is to perform a uniform rotation $e^{-i\theta J_y}$ for some angle $\theta$ and measure $J_z$. If $m= m_t$ is measured, we are done; otherwise, we iterate this procedure, choosing subsequent rotation angles $\theta$ based on the prior outcome of the measurement of $J_z$.

A natural strategy for choosing the rotation angles $\theta$ is to maximize the overlap of the current state with the target Dicke state on each iteration of the algorithm. Since both the initial state and the state following each collective $J_z$ measurement are of the form $\ket{j,m}$, the task is to choose $\theta = \theta_{m_t,m}$ to maximize the quantity
\begin{equation}
\label{eq:transition-prob}
    \abs{d^j_{m_t, m}(\theta_{m_t,m})}^2 = \abs{\braket{j,m_t|e^{-i\theta_{m_t,m}J_y}|j,m}}^2,
\end{equation}
where $d^j_{m_t, m}(\theta)$ are known as elements of the Wigner $d$-matrix. While the functional form of these matrices is known, it is practically difficult to optimize for arbitrary $j$, $m$, and $m_t$ without resorting to numerical methods. However, we identify a relatively simple analytic choice of angle that is nearly optimal---in the sense that we numerically observe only a constant overhead in time compared to the optimal angle---and for which we provide a rigorous analysis of the running time.

To find a suitable set of angles $\theta_{m_t,m}$ and to understand the properties of our algorithm, we find it useful to first visualize the Dicke states in terms of their phase space distributions. In particular, consider the Husimi-Q distribution $Q(\Omega) = \abs{\braket{\psi|\Omega}}^2/\pi$ for a state $\ket{\psi}$, where $|\Omega\rangle$ is a coherent spin state oriented along an axis $\mathbf{n}$ with polar and azimuthal angles $\Omega = (\theta,\phi)$~\cite{arecchi_atomic_1972}. In terms of the Dicke states, these are defined via
\(
\ket{\theta,\phi} = e^{-i\phi J_z}e^{-i\theta J_y}\ket{j,j}
\).
We can thus geometrically map the Q distribution for a state onto the surface of a collective Bloch sphere with radius $\sqrt{j(j+1)}$. The Q distribution for the Dicke state $\ket{\psi}=\ket{j,m}$ is
\(    Q_m(\Omega) = \frac{1}{\pi} \binom{2j}{j+m} \cos^{2(j+m)}\tfrac{\theta}{2}\sin^{2(j-m)}\tfrac{\theta}{2}.
\)
On the collective Bloch sphere surface, as $j\rightarrow \infty$, these correspond to narrow horizontal rings of radius $r_m = \sqrt{j(j+1) - m^2}$ located at a height $m$ above the equator for $m\neq \abs{j}$; for $m = \pm j$, this instead corresponds to a narrow Gaussian distribution at either pole~\cite{agarwal_relation_1981,dowling_wigner_1994}. For a uniform rotation of $\ket{j,m}$ via $e^{-i\theta J_y}$, the Q distribution undergoes the analogous rotation on the collective Bloch sphere. This results in a ``tilted ring'' distribution, where the normal vector to the plane of the ring forms an angle $\theta$ to the z-axis.

Within this geometric picture, a reasonable choice of angle $\theta_{m_t,m}$ is one that maximizes the overlap of the Q distributions of the rotated state and the target Dicke state in the limit of large $j$. This maximum occurs when the corresponding ring distributions intersect at a point sharing the same tangent vector, as shown in Fig~\ref{fig:geometric-procedure}. 
This condition is met for rotation angle
\begin{align}
\label{eq:angle}
    \theta_{m_t,m} = \arcsin{
    [(mr_{m_t} - m_tr_m)/r_0^2]
    }.
\end{align}

In the End Matter, we argue that our algorithm generically prepares arbitrary target Dicke states in time $O(\log (j-m_t))$. In particular, we predict that the running time decreases with increasing $m_t$, which is consistent with numerical calculations discussed below.

For the remainder of this paper, we focus on the case $m_t = 0$, corresponding to the Dicke state with maximal interspin entanglement. Then Eq.~\eqref{eq:angle} reduces to $\theta_{0,m} = \arcsin\parens[\big]{m/\sqrt{j(j+1)}}$. For simplicity and with negligible impact, we approximate this as $\theta_m \coloneq \arcsin(m/j)$. We see in Fig.~\ref{fig:main}(a) that our choice of $\theta_m$ is numerically close to the optimal $\theta_m^*$ that maximizes Eq.~\eqref{eq:transition-prob}.

As shown in Fig.~\ref{fig:main}(b) and its inset, there are some $m$ for which $\abs*{d^j_{0,m}(\theta_m)}^2 < \abs*{d^j_{0,j}(\theta_j)}^2$, i.e., we have a lower probability of reaching the $m=0$ state with the optimal rotation than if we start over with $m=j$ and rotate by $\frac{\pi}{2}$. In these cases, we choose to reset all qubits to $\ket{0}$ and restart from $m=j$. Empirically, we observe this to hold for $\abs{m} \gtrsim j^{3/4}$. Though there is negligible difference in the numerical runtime, we include this reset whenever $\abs{m} > j^{1/2}$ for ease of the formal proof in the Supplement. The final procedure is Algorithm \ref{alg}.

\begin{algorithm}[H]
\caption{Preparation of the $m_t = 0$ Dicke state using global rotations and $J_z$ measurements}
\label{alg}
\begin{algorithmic}[1]
\State Initialize $2j$ qubits each to $\ket{0}$, $m = j$
\While {$m \ne 0$}
	\State Rotate by $\exp\parens{-i\theta_{m} J_y}$
	\State $m \gets$ measure ${J_z}$
	\If {$\abs{m} > \sqrt{j}$}
		\State Reset all qubits to $\ket{0}$, $m = j$
	\EndIf
\EndWhile
\end{algorithmic}
\end{algorithm}

\textit{Runtime analysis.---}In this section we sketch the proof of the main result. The full calculations are provided in the Supplemental Material. We consider each iteration of the \textbf{while} loop to take unit time. 

\begin{theorem}
\label{thm:main}
Algorithm \ref{alg} prepares the Dicke state $\ket{j, 0}$ in expected $O(\log j)$ time.
\end{theorem}

At a high level, our strategy is to show that $\ev{\abs{m(t)}^\alpha}$, the expectation of $\abs{m}^\alpha$ at time $t$ for some constant $\alpha >  0$, decays to $0$ exponentially in $t$. This means that, given any $\eps > 0$, we can achieve $\ev{\abs{m(t)}^\alpha} < \eps$ within logarithmic time. Since the values of $m$ are discretized, if $\ev{\abs{m(t)}^\alpha} < \eps$, then the probability of being in the $m=0$ state after $t$ steps is $\Pr[m(t) = 0] > 1-\eps$, as shown in the supplement.

First, we show that starting from the initial $m=j$ state, we can obtain $\abs{m} \le \sqrt{j}$ in expected $O(1)$ time. This follows since the measurement outcomes are binomially distributed around $0$, and we obtain a measurement within a standard deviation with constant probability.

Next, we use the stationary phase approximation to show that, for $1 \ll m \le \sqrt{j}$, we have the following asymptotic expansion of the Wigner $d$-matrix element:
\begin{align}
d^j_{m',m} & (\theta_m) = \sqrt{\frac{2}{\pi m}} \bracks{1 - (1-x)^2}^{-1/4} \times \nonumber\\
& \cos\bracks{m(1\!-\!x) \arccos(1\!-\!x) - m \sqrt{1\!-\!(1\!-\!x)^2} + \tfrac{\pi}{4}} \nonumber\\
&+ O(\max\{{m}^2 /j^{2},  1/{m} j \})
\label{eq:d-asymp}
\end{align}
for $0 < m' < 2m$ where $x = m'/m$, and the matrix element is negligible for other $m'$.

We do not have asymptotic expressions for the Wigner $d$-matrix elements for $m > \sqrt{j}$, but these are not needed as we simply reset if we measure $m > \sqrt{j}$. However, resetting to $m=j$ drastically increases the expectation of $\abs{m}$ at that time step. Nevertheless, at the next time step we expect to recover a state with $m \le \sqrt{j}$ with constant probability. Thus, we introduce a proxy variable $M \coloneq \min(\abs{m}, \sqrt{j}+1)$. This still has the desirable property that if $\ev{M(t)^\alpha} < \eps$, then the probability of reaching $m=0$ after $t$ steps is at least $1 - \eps$.

Using the asymptotic expansion in Eq.~\eqref{eq:d-asymp}, we show that there exists a constant $c < 1$ such that, for every $m$, we have
\(
\sum_{m'=0}^{2m} \abs*{d^j_{m',m}(\theta_m)}^2 \frac{M^\alpha}{M'^\alpha} < c,
\)
where $M'$ is defined as a proxy for $m'$ analogously to $M$.
This implies that $\frac{\ev{M(t+1)^\alpha}}{\ev{M(t)^\alpha}} < c$ for each $t$, so $\ev{M(t)^\alpha} < c^t \ev{M(0)^\alpha} = c^t (\sqrt{j}+1)^\alpha$.
Therefore, for any desired $\eps > 0$, we can attain $\ev{M(t)^\alpha} < \eps$ with
\begin{equation}
t = \frac{\alpha\log(\sqrt{j}+1) + \log(1/\eps)}{\log(1/c)} = O(\log j)
\end{equation}
steps, as claimed.

\textit{Numerics.---}Our algorithm can be understood as a discrete-time Markov chain with $2j+1$ states corresponding to $m \in \{-j, -j+1, \dots, j\}$. The transition probabilities
\(
\Pr[m \to m'] = \abs*{d_{m',m}^j(\theta_m)}^2
\)
can be arranged into a stochastic matrix $P$, where $P_{ab}$ is the probability of transitioning from the $a$th to the $b$th state. A visualization of $P$ is shown in Fig.~\ref{fig:main}(b).

This is an absorbing Markov chain with the single absorbing state $m = 0$. The average number of steps before absorption can be calculated directly from $P$~\cite{kemeny1960finite}. Figure~\ref{fig:main}(c) numerically compares the performance of Algorithm \ref{alg} with variations in the choice of angles. In particular, our geometrically motivated angles $\theta_m = \arcsin\parens{ {m}/{j}}$ perform slightly worse than the optimal angles $\theta_m^*$, but exhibit the same logarithmic scaling for the expected number of steps, with a relatively small constant prefactor.

Finally, in Fig.~\ref{fig:main}(d), we examine the preparation of Dicke states with arbitrary $m_t$, utilizing our choice of angles $\theta_{m,m_t}$ in Eq.~\eqref{eq:angle}. For various fixed $j$, we observe that this choice of angle results in an average number of steps strictly less than that required for the $m_t=0$ case. We argue in the End Matter that this behavior is expected, and that the corresponding average number of steps scales as $O(\log(j-m_t))$.

\begin{figure}[t]
\centering
\includegraphics[width = 0.49\textwidth]{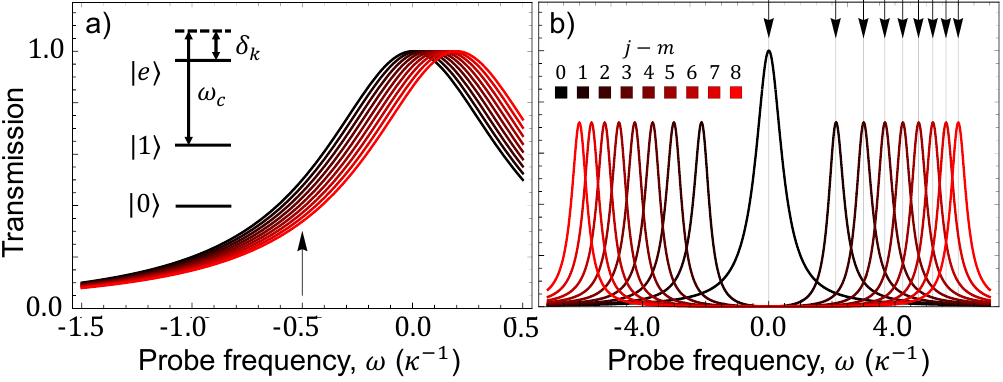}
\caption{Two schemes to experimentally implement Hamming weight measurements, illustrated for $n=8$ qubits. (a) Probe on the side-of-fringe in the dispersive cavity regime, with photon counting in transmission to determine cavity shift magnitude (inset: level diagram for our detection scheme). (b) Probe simultaneously at each possible resonance frequency in the resonant cavity regime, with heterodyne transmission detection to determine which frequency tone transmits.}
\label{fig:measurement}
\end{figure}

\textit{Measurement methods.---}Collective Hamming weight measurements may be directly implemented on an ensemble of $n$ atomic qubits in which one of the two qubit states is coupled to a single-mode cavity. We illustrate this by considering three-level atoms as depicted in Fig.~\ref{fig:measurement}(a), with states $\{|0\rangle, |1\rangle, |e\rangle\}$, where $|0\rangle$ and $|1\rangle$ are the computational subspace and the cavity dispersively couples $|1\rangle$ to $|e\rangle$, i.e. with a large detuning between the cavity frequency and the atomic transition frequency between $|1\rangle$ and $|e\rangle$. In the rotating frame of a bare cavity photon (at lab frame frequency $\omega_c$), the atom-cavity interaction Hamiltonian is
\begin{equation}
H = a^\dagger a \sum_{k=1}^n \frac{g^2}{\delta_k} \mathcal{P}_k
\end{equation}
where $a^\dagger$ ($a$) creates (annihilates) a cavity photon, $2g$ is the single-photon Rabi frequency, $\delta_k$ is detuning of the $k$th atom from the cavity, and $\mathcal{P}_k = \half (I - 2(\mathbf{S}_k)_z)$. For our application, we set all detunings to be equal, i.e.~$\frac{g^2}{\delta_k} = \chi$ for all $k$ for some constant $\chi$. Thus, the total cavity shift is $\Delta_a = \chi \mathcal{P}$, where $\mathcal{P} = \sum_k \mathcal{P}_k$, and the lab-frame cavity transmission spectrum is
\(
T(\omega) = \frac{\kappa^2}{(\omega - \omega_c - \Delta_a)^2 + \kappa^2},
\)
where $\kappa$ is the cavity linewidth and $\omega$ is the angular frequency of the cavity probe.

Assuming that the maximum total cavity shift is sufficiently small, i.e., $\chi n\ll\kappa$, we can probe on the side of the transmission peak, as depicted in Fig.~\ref{fig:measurement}(a), taking $\omega - \omega_c = \kappa$, to yield a Fisher information of
\(
\mathcal{I}(\Delta_a) = \frac{4\kappa^2}{(2\kappa^2 - 2\kappa\Delta_a + \Delta_a^2)^2}
\)
for a single photon.
The Cramer-Rao bound then gives
\(
\Var(\widetilde{\Delta}_a) \ge \frac{1}{\mathcal{I}(\Delta_a)} = \parens{\kappa - \Delta_a + \frac{\Delta_a^2}{2\kappa}}^2 \sim \kappa^2
\label{eq:single-fisher-variance}
\)
for any unbiased estimator $\widetilde{\Delta}_a$ of $\Delta_a$. 
Averaging over $N$ photons and taking $\mathcal{P} = \frac{\Delta_a}{\chi}$ gives
\(
\Var(\widetilde{P}) \sim \frac{1}{N} \parens{\frac{\kappa}{\chi}}^2.
\)
We ensure that this variance is $O(1)$ by taking $N \sim \parens{\frac{\kappa}{\chi}}^2$ photons.

This scheme is straightforward, and the assumption of small total cavity shift is easy to satisfy experimentally. However, the number of photons required is quite large because the differential signal between possible Hamming weight measurements is, by assumption, small.

It is more experimentally advantageous to maximize the differential signal regardless of atom number, maximizing the resolvability of neighboring Hamming weights. Therefore we may place the cavity on resonance with the $|1\rangle \leftrightarrow |e\rangle$ transition, which yields the bare cavity transmission spectrum if all atoms are in $|0\rangle$ and vacuum Rabi spectra with splitting $2g\sqrt{n_1}$ for Hamming weights $n_1$. The task of Hamming weight measurement is then the task of determining which of the possible spectra is realized by the cavity.

To this end, we propose a multichromatic probe laser with a spectral peak at each of the possible vacuum Rabi resonances, as depicted in Fig.~\ref{fig:measurement}(b), and power in each peak chosen to yield equal \textit{transmitted} photon number. In the fully resolved limit, this laser probes each Hamming weight possibility in a time and with a number of transmitted photons which is independent of the total atom number. The frequency of the cavity transmission signal, which carries the desired Hamming weight information, may be revealed using standard optical heterodyne techniques.

The principal cost of this scheme is the requirement that the atom-cavity coupling $g$ is $\Omega(\sqrt{n})$, to be able to resolve neighboring peaks that are $g \sqrt{n} - g \sqrt{n-1} \sim g/\sqrt{n}$ apart. Additionally, while not a fundamental limitation, the number of tones and bandwidth to produce the multichromatic probe laser and perform optical heterodyne measurement scale linearly with $n$.

\textit{Discussion.---}In this paper, we have shown an algorithm for preparing Dicke states with depth and width logarithmic in the number of qubits. The algorithm is compatible with existing experimental platforms, using only sequences of global single-qubit rotations and collective Hamming weight measurements. We have proposed an experimental framework in which the collective measurements can be performed in constant depth, leading to a $\log$-depth circuit. Even in the absence of collective measurements, using existing Hamming weight protocols gives a circuit for preparing Dicke states in polylog depth, outperforming several recent works in the regime where the desired weight is linear in the number of qubits. This is also the first such algorithm that only uses logarithmically many ancilla qubits.

Our results illustrate the power of utilizing a phase space approach based on the Q distribution, which affords an intuitive geometric understanding of our algorithm and, crucially, provides an effective choice of rotation angles for our algorithm. These nearly optimal parameters alleviate the need for any numerical optimization in our approach, and also enable an analytic study of the asymptotic properties of our algorithm.

The cavity system used to implement our collective measurements also holds potential for further applications.
The small number of expected iterations makes this suitable for near-term implementation. Even if the qubits decohere, it is inexpensive to simply reset the experiment and retry.

In this work, we were only concerned with Hamming weight measurements, which arise from setting all detunings to be equal. A natural extension is to relax the assumption that all atoms couple to the cavity equally, allowing for a richer class of measurements.
For example, we may obtain a superposition of Dicke states $\ket{j, m} + \ket{j, -m}$ (which has metrological applications~\cite{lin_covariant_2024}) by probing at the midpoint frequency.

\begin{acknowledgements}

\textit{Acknowledgements.---}We thank Jacob Lin for helpful discussions. J.Y., A.M.C., and A.V.G.~were supported in part by the DoE ASCR Quantum Testbed Pathfinder program (awards No.~DE-SC0019040 and No.~DE-SC0024220), DoE ASCR Accelerated Research in Quantum Computing program (awards No.~DE-SC0020312 and No.~DE-SC0025341), and NSF QLCI (award No.~OMA-2120757). J.Y.~and A.V.G.~were also supported in part by DARPA SAVaNT ADVENT, NSF STAQ program, and AFOSR MURI. J.Y.,~and A.V.G.~also acknowledge support from the U.S.~Department of Energy, Office of Science, National Quantum Information Science Research Centers, Quantum Systems Accelerator. S.R.M. is supported by the NSF QLCI grant OMA-2120757. Y.-X.W.~acknowledges support from a QuICS Hartree Postdoctoral Fellowship. N.S.~acknowledges support, in part, from ARO (W911NF-24-1-0064) and ARL (W911NF-24-2-0107).

\end{acknowledgements}

\bibliography{refs}

\newpage
\onecolumngrid

\appendix
\newpage

\section{End Matter---Geometric argument}\label{app:geometric}

In this End Matter, we give the details of a geometric argument that Algorithm \ref{alg:general} reaches the $\ket{j, m_t}$ Dicke state in $O(\log(j-m_t))$ iterations. We let $m_t \ge 0$ without loss of generality as the $m_t \le 0$ case is symmetric via a global flip. For simplicity, we do not include a reset condition analogous to that in Algorithm~\ref{alg} as it is not needed for our argument.

\begin{algorithm}[H]
\caption{Preparation of arbitrary $\ket{j, m_t}$ Dicke state using global rotations and $J_z$ measurements}
\label{alg:general}
\begin{algorithmic}[1]
\State Initialize $2j$ qubits each to $\ket{0}$, $m = j$
\While {$m \ne m_t$}
	\State Rotate by $\exp\parens{-i\theta_{m_t, m} J_y}$, where $\theta_{m_t,m} = \arcsin{
    [(mr_{m_t} - m_tr_m)/r_0^2]
    }$
	\State $m \gets$ measure ${J_z}$
\EndWhile
\end{algorithmic}
\end{algorithm}

We describe the expected behavior within the geometric picture, whose relevant quantities are shown in Fig. \ref{fig:geometry}(a). The strategy is similar to the proof of Theorem \ref{thm:main}. First, we use the geometric model to derive a coarse-grained expression for the transition probabilities in the large-$j$ regime. Within this model, we compute the transition probabilities $\Pr[m \to m']$ as the overlap of the Q distribution ring of $\ket{j, m'}$ with the rotated ring from $\ket{j, m}$. Then we show that, for some constant $\alpha > 0$, the expected deviation $\ev{(m'-m_t)^\alpha}$ decays exponentially in time. The conclusion then follows in the same manner.

\begin{figure}[b]
\includegraphics[width=0.6\textwidth]{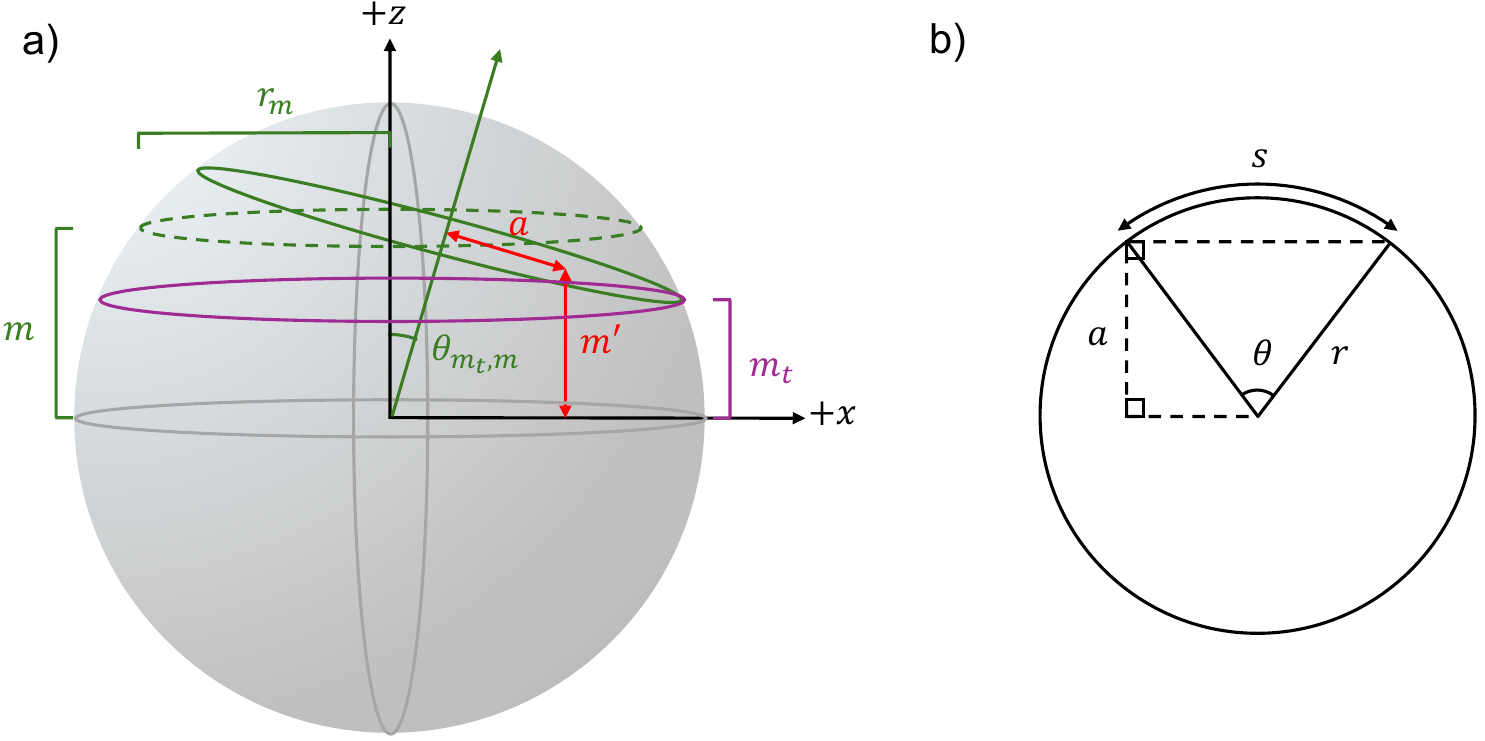}
\caption{(a) Geometry of the (rotated) Dicke states, idealized as rings on the collective Bloch sphere, corresponding to the large-$j$ limit of the Husimi-Q distribution. We show the ring corresponding to $|j,m\rangle$ (green) before (dotted) and after (solid) a rotation by an angle $\theta_{m_t,m}$. We also show the ring for the corresponding target state $|j,m_t\rangle$ (purple). Other variables $m'$ and $a$ used in our calculation of the probability distribution function are shown in red. (b) A two-dimensional cross section of a Dicke ring with the relevant parameters for calculating arc length.}
\label{fig:geometry}
\end{figure}

We begin by computing the arc length of a circular segment a distance $a$ from the center of a ring of radius $r$, as shown in Fig.~\ref{fig:geometry}(b). This is equivalently described by the arc length of the circular sector of angular width $\theta$, defined such that $\cos(\theta/2) = a/r$. The corresponding arc length is then $s = 2r \arccos(a/r)$, and we have the (unsigned) infinitesimal arc length
\begin{align}
    \dd{s} = \frac{2 \dd{a}}{\sqrt{1-(a/r)^2}}.
\end{align}
Now, consider a horizontal ring at height $m$ above the origin (with radius $r_m$), and rotated by an angle $\theta_{m_t,m}$ about $+y$ so that the bottom edge of the rotated ring lies at height $m_t$; $\theta_{m_t,m}$ is defined via Eq.~\eqref{eq:angle}. For the differential arc length of this rotated ring, we have the relation $a = (m'-m_t)/\sin\theta_{m_t,m} - r_m$, so
\begin{align}
    \dd{s} = \frac{2r_m \dd{m'}}{(m' - m_t) \sqrt{2r_m\sin\theta_{m_t,m}/(m'-m_t) - 1}}.
\end{align}
This expression is defined in the range $m_t < m' < m_t + 2r_m \sin \theta_{m_t, m}$.

In the large $j$ limit, where the Q-distribution of the state $e^{-i\theta_{m_t,m}J_y}|j,m\rangle$ is well represented by this tilted ring, we assume that the probability to obtain $J_z = m'$ is proportional to the arc length  lying between $m'$ and $m'+dm'$. Properly normalizing, we thus have the continuous probability distribution function (pdf)
\begin{align}
    p(m',m;m_t) = \frac{1}{\pi r_m\sin\theta_{m_t,m} \left\{
    \left[\left(m'-m_t\right)/\left(r_m\sin\theta_{m_t,m}\right)\right]\left[2 - \left(m'-m_t\right)/\left(r_m\sin\theta_{m_t,m}\right)\right]\right\}^{1/2}}
    \label{eq:pdf}
\end{align}
for $m_t < m' < m_t + 2r_m \sin \theta_{m_t, m}$, and assume zero probability for $m'$ outside this range. We note that while the integral of this expression converges, the pdf diverges at $m_t$ and $m_t + 2r_m$. This indicates that the overlap probability is maximal at these points, and we thus utilize the set of rotation angles $\theta_{m_t,m}$ for preparing target Dicke states $|j,m_t\rangle$.

This offers a coarse-grained way to predict moments of $J_z$. In particular, Eq.~\eqref{eq:pdf} takes the form of a beta distribution, with moments
\begin{align}
    \ev{(m'-m_t)^\alpha \mid m} &= \int_{m_t}^{m_t + 2r_m\sin\theta_{m_t,m}}(m'-m_t)^\alpha p(m',m;m_t) \dd{m'} \\
	&= \frac{B(\alpha + 1/2,1/2)}{\pi}\left(2r_m\sin\theta_{m_t,m}\right)^\alpha \label{eq:beta_geo}
\end{align}
for beta function $B(x,y) = \Gamma(x)\Gamma(y)/\Gamma(x+y)$ and gamma function $\Gamma(x)$.
Here, $\ev{(m'-m_t)^\alpha \mid m}$ is the expectation of $(m'-m_t)^\alpha$, conditional on the previous state being $\ket{j,m}$.

For $0<\alpha< 1$, we have $\ev{(m'-m_t)^\alpha \mid m} < c (m-m_t)^\alpha$ for some $j$-independent constant $0 < c < 1$. To see this, we first note that 
\begin{align}
    \frac{B(\alpha + 1/2,1/2)}{\pi}\left(\frac{2r_m\sin\theta_{m_t,m}}{m-m_t}\right)^\alpha \leq \frac{B(\alpha + 1/2,1/2)}{\pi}2^\alpha,\label{eq:beta_geo_upperbound}
\end{align}
using the fact that $r_0 \leq r_{m_t} \leq r_m$ for all $m$ such that $m \geq m_t$, $m \leq j$. Now let $f(\alpha) = \ln[ B(\alpha+1/2,1/2)2^\alpha/\pi]$. We have $f''(\alpha)= \psi^{(1)}(\alpha+1/2) - \psi^{(1)}(\alpha+1)$, where $\psi^{(n)}(\alpha) = d^{(n+1)} \ln\Gamma(\alpha)/d^{(n+1)}\alpha$ is the polygamma function of order $n$. From the series representation $\psi^{(n)}(\alpha) = (-1)^{n+1}n!\sum_{k=0}^\infty (\alpha+k)^{-(n+1)}$~\cite{abramowitz_handbook_1972}, we see that $\psi^{(1)}(\alpha)$ is strictly decreasing for all $\alpha > 0$, since $\psi^{(2)}(\alpha) < 0$. This implies that $f''(\alpha) > 0$, so $f(\alpha)$ is convex for $\alpha > 0$. Now, $f(0) = 0$ and $f(1)=0$, so for $0 <\alpha<1$ we have $f(\alpha) < 0$ and the right-hand side of Eq.~\eqref{eq:beta_geo_upperbound} is strictly upper bounded by $1$. For any fixed $\alpha$ in this range, we may therefore select a ($j$-independent) constant $c<1$ such that Eq.~\eqref{eq:beta_geo} is strictly upper bounded by $c(m-m_t)^\alpha$.

Thus $\ev{(m'-m_t)^\alpha}$ decays by a factor of $c$ at each step of Algorithm~\ref{alg:general}. Let $\ev{(m'-m_t)^\alpha}_k$ denote the expectation after $k$ steps. We initially have $\ev{(m'-m_t)^\alpha}_0 = (j-m_t)^\alpha$, so by induction $\ev{(m' - m_t)^\alpha}_k < c^k (j-m_t)^\alpha$. Now, if we require $\braket{(m-m_t)^\alpha}_n < \varepsilon$ for some $\varepsilon > 0$, this can be achieved in
\begin{align}
    n > \frac{\alpha\ln (j-m_t) + \ln (1/\varepsilon)}{\ln(1/c)}
\end{align}
steps, or $n = O(\log (j-m_t))$, as claimed.

\newpage
\section{Supplemental Material: Asymptotic expansions of Wigner $d$-matrices}

In this Supplement, we rigorously analyze the time to prepare the $m_t = 0$ Dicke state. We proceed via three lemmas.

\setcounter{lemma}{0}

\begin{lemma}
\label{lem:sqrt.j.step}
Starting from $m=j$, we can obtain a state with $\abs{m' } \le \sqrt{j}$ in expected $O(1)$ time.
\end{lemma}

\begin{proof}
From $m = j$, we rotate by $\theta_j = \frac{\pi}{2}$, which gives the state
\begin{equation}
e^{-i (\pi/2) J_y} \ket{j, j} = \parens{\frac{\ket{0} + \ket{1}}{\sqrt{2}}}^{\otimes n} = \frac{1}{2^{n/2}} \sum_{x \in \{0, 1\}^n} \ket{x}.
\end{equation}
The probability of measuring weight $w$ is ${\binom{n}{w}}/{2^n}$, which is a binomial distribution with mean $\frac{n}{2}$ and variance $\frac{n}{4}$. If $\abs{m'} > \sqrt{j}$, 
then the difference between $w$ and the mean is of the same order as the standard deviation, so that the probabilities of those instances sum up to a bounded $O(1)$ probability, independent of $n$. In such a case, we simply reset and try again. The expected number of attempts to succeed is $O(1)$.
\end{proof}

\begin{lemma}
\label{lem:d-asymptotic}
In the regime $m= \omega(1)$, $m = O( \sqrt{ j})$, and $0 <m' < 2m$, we have
\begin{equation}
\label{neq:wignerD.asymp.eff}
 d ^{ j } _{m' m} (\beta_{m} ) =
\sqrt{\frac{2}{\pi m} } 
\frac{\cos \left[ m (1-x) \arccos (1-x) - m\sqrt{1- (1 -x) ^{2}} +\frac{\pi}{4}\right ] }{  [1 - (1 -x) ^{2}]^{\frac{1}{4}} }
+ O(\max\{{m}^2 j^{-2},  {m}^{-1} j^{-1}\}),
\end{equation}
where $\beta _{m} = \arcsin(m/j) $ and $x = m'/m$.
\end{lemma}

\begin{proof} 
We start by rewriting the Wigner $d$-matrix element as the following integral (see Eq.~(11) in~\cite{WignerDreview1988}):
\begin{align}
d ^{ j } _{m' m} (\beta _{m}) &= \frac{(-1)^{m'-m} }{2 \pi}
\left [ \frac{(j+m')! (j-m')! }{(j+m)! (j-m)! }\right ] ^{\frac{1}{2}}
\nonumber \\ 
& \quad \times \int ^{2\pi} _{0} \left ( e ^{i\frac{\phi}{2}} \cos \frac{\beta _{m}}{2} + i e ^{-i\frac{\phi}{2}} \sin \frac{\beta_{m}}{2} \right ) ^{j-m} 
\left ( e ^{-i\frac{\phi}{2}} \cos \frac{\beta _{m}}{2} + i e ^{i\frac{\phi}{2}} \sin \frac{\beta _{m}}{2} \right ) ^{j+m} 
e ^{i m' \phi} \dd{\phi}
.
\end{align}
By shifting the integration variable $\phi \to \phi + \frac{\pi}{2}$, we can equivalently write this as 
\begin{align}
d ^{ j } _{m' m} (\beta _{m}) &= \frac{(-1)^{m'-m} }{2 \pi} 
\left [ \frac{(j+m')! (j-m')! }{(j+m)! (j-m)! }\right ] ^{\frac{1}{2}} 
\nonumber \\ 
\label{neq:wignerD.int.gen}
& \quad \times \int ^{2\pi} _{0} \left ( \cos \frac{\beta _{m}}{2} +  e ^{-i\phi } \sin \frac{\beta _{m}}{2} \right ) ^{j-m} 
\left ( \cos \frac{\beta _{m}}{2} - e ^{i\phi} \sin \frac{\beta _{m}}{2} \right ) ^{j+m} 
e ^{i (m'-m) \phi}  \dd{\phi}
.
\end{align}
We now focus on the integrand, which can be split into the product of a positive magnitude and a phase part: 
\begin{align}
\label{neq:wignerD.split}
d ^{ j } _{m' m} (\beta _{m}) &= \frac{(-1)^{m'-m} }{2 \pi} 
\left [ \frac{(j+m')! (j-m')! }{(j+m)! (j-m)! }\right ] ^{\frac{1}{2}}   \int ^{2\pi} _{0} g(\phi) e ^{i f(\phi) }  \dd{\phi}
, \\
g(\phi) &= \exp \left [ \frac{j-m}{2} \log (1+\sin \beta_{m}\cos \phi ) + \frac{j+m}{2} \log (1- \sin \beta _{m}\cos \phi )\right]
, \\
f (\phi) &=  (m'-m) \phi -  (j-m) \arctan \frac{ \sin \phi \tan \frac{\beta _{m}}{2}}{1+ \cos \phi \tan \frac{\beta _{m}}{2}}  -  (j+m) \arctan \frac{ \sin \phi \tan \frac{\beta _{m}}{2}}{1 - \cos \phi \tan \frac{\beta _{m}}{2}} 
. 
\end{align}
We choose the rotation angle $\beta _{m} = \arcsin (m/j)$ according to the optimal angle predicted by the geometric picture, as discussed in the main text.

In order to make further approximations in the $j \to \infty$ limit, we first consider the derivatives of the functions $g(\phi) $ and $f (\phi) $ with respect to $\phi$:
\begin{align}
& \frac{\dd{g(\phi)}}{\dd{\phi}} = - g(\phi ) \sin \beta _{m}\sin \phi  \frac{j \sin \beta _{m}\cos \phi +m}{1- \sin ^{2}\beta _{m}\cos ^{2} \phi }  
, \\
& \frac{\dd{f(\phi)}}{\dd{\phi}} = - (m-m') - \frac{(j-m)\tan \frac{\beta _{m}}{2} (\cos \phi  + \tan \frac{\beta _{m}}{2}) }{ 1+ (\tan \frac{\beta _{m}}{2} )^{2} +2\cos \phi \tan \frac{\beta _{m}}{2} } 
- \frac{(j+m)\tan \frac{\beta_{m}}{2} (\cos \phi  - \tan \frac{\beta _{m}}{2}) }{ 1+ (\tan \frac{\beta_{m}}{2} )^{2} -2\cos \phi \tan \frac{\beta_{m}}{2} } 
. 
\end{align}

For $m, m' = \omega(1)$ and $m = O( \sqrt{j})$, the magnitude function $g(\phi)$ varies much slower (by an extra factor of $\sin \beta _{m} = m/j$) relative to the phase function $f (\phi)$, and we can apply the stationary phase approximation to the integral in Eq.~\eqref{neq:wignerD.int.gen}. In this regime, noting that $\sin \beta _{m}= m/j$, the derivative of the phase function can be well approximated as
\begin{align}
\label{neq:st.phase.deriv}
\frac{\dd{f(\phi)}}{\dd{\phi}} = - (m-m')  - m \cos \phi   [1+
O({m}^2 j^{-2})], 
\end{align}
so that, under the stationary phase approximation (see Sec.~3 of Ref.~\cite{stationary_phase_book_2001}), we have that in the $j \to\infty$ limit, with $m, m' = \omega(1)$ and $m = O( \sqrt{j})$, the following equation holds: 
\begin{align}
\label{neq:WignerD.sph.approx}
\lim _{\substack{j \to\infty \\ m, m' = \omega(1)\\ m = O( \sqrt{j})}}
\int ^{2\pi} _{0} &  g(\phi) e ^{i f(\phi) }  \dd{\phi}
= \sum _{\phi _{s} = \pi \pm \arccos \frac{m-m'}{m}}
g(\phi _{s}) e ^{i f(\phi _{s} ) \pm i\frac{\pi}{4} }  \sqrt{\frac{2\pi }{\left |\frac{d^{2} f(\phi )}{d \phi ^{2} } \right | _{\phi = \phi_{s}} }}
+ O(\max\{{m}^2 j^{-2},  {m}^{-1} j^{-1}\}).
\end{align}
Substituting Eq.~\eqref{neq:WignerD.sph.approx} back into Eq.~\eqref{neq:wignerD.split}, we obtain the desired approximate asymptotic expression for the Wigner $d$-matrix in this regime:
\begin{align}
\label{neq:wignerD.asymp.final}
\lim _{\substack{j \to\infty \\ m, m' = \omega(1)\\ m = O( \sqrt{j})}}
d ^{ j } _{m' m} (\beta _{m} ) & = \frac{\sqrt{\frac{2}{\pi} } \left [\frac{j^{2}-{m'}^{2}}{j^{2}-{m}^{2}} \right] ^{\frac{1}{4}} }{[{m}^{2} - (m -m') ^{2}]^{\frac{1}{4}} } \cos \left[ (m -m') \arccos \frac{m-m'}{m } - \sqrt{{m}^{2}- (m -m') ^{2}} +\frac{\pi}{4}\right ]
\nonumber \\
& \quad + O(\max\{{m}^2 j^{-2},  {m}^{-1} j^{-1}\}).
\end{align}
Rewriting with $x = m'/m$ gives the result.
\end{proof}

Note that, if $m' < 0$ or $m' > 2m$, then Eq.~\eqref{neq:st.phase.deriv} does not have a solution for $\phi \in [0, 2\pi]$, meaning that the $d$-matrix element is negligible up to leading order.

As a corollary of Lemma~\ref{lem:d-asymptotic}, in the regime where $m = \omega(1)$ and $\abs{m' - m} = O(1)$, we have
\begin{equation}
\label{neq:wignerd.bessel}
\lim _{j\to\infty} d_{m'm}^j(\beta _{m}) = J_{m-m'}(m)  ,
\end{equation}
where $J _{\ell} (x)$ is the Bessel function of the first kind. As the zeros of the Bessel function are transcendental \cite{Siegel1929}, we see that the asymptotic transition probabilities from $m$ to $m'$ are nonzero in this regime with $|d_{m'm}^j(\beta _{m})|^{2} = \Theta (1)$.

Define the random variable $M$ by
\begin{equation}
\label{neq:def.M.via.m}
M  = \begin{cases} m & \abs{m} \le \sqrt{j} \\ \sqrt{j}+1 & \abs{m} > \sqrt{j}. \end{cases}
\end{equation}
This is a proxy for $m$. We also introduce $M'$ as a function of $m'$ in a similar fashion, as the reset drastically increases $m'$ to $j$, making the expectations of $\ev{{m'}^\alpha}$ suboptimal for the runtime analysis.

\begin{lemma}
\label{lem:ratio-small}
There exists a constant $c < 1$ and positive exponent $0<\alpha<1$ such that, for every $m = \omega (1)$ and  $m \le \sqrt{j}$,
\begin{equation}
\sum_{m'} \abs{d ^{ j } _{m' m} (\beta_{m} )}^2 \frac{{M '}^\alpha }{M ^\alpha } < c.
\end{equation}
\end{lemma}

\begin{proof}
Now consider
\begin{align} 
\label{neq:MC.crit.ratio.sep}
& \sum _{m' } \mathcal{P} (m \to m ' )  \frac{{M'}^{\alpha} }{M^{\alpha} } 
= \frac{ (1- \sum _{m '<\sqrt{j} } |d ^{ j } _{m 'm} (\beta _{m  } ) | ^{2}  ) (\sqrt{j}+1)^{\alpha} + \sum _{m' <\sqrt{j} } |d ^{ j } _{m' m} (\beta _{m  } ) | ^{2}  {m'}^{\alpha} }{m^{\alpha} } 
. 
\end{align}
We can make use of the asymptotic expression in Eq.~\eqref{neq:wignerD.asymp.final} to simplify this expression. We first note that, for $j \to\infty$ with $m, m' = \omega(1)$ and $m = O( \sqrt{j})$, the reset probability is given by 
\begin{align}
& 1- \sum _{m' <\sqrt{j} } |d ^{ j } _{m' m} (\beta _{m  } ) | ^{2} 
\nonumber \\
&\quad= 1- \frac{1}{\pi} \sum  _{m' <\sqrt{j} }  \left \{
\frac{ 1+ \sin \left[ 2 (m -m') \arccos \frac{m-m'}{m } - 2 \sqrt{{m}^{2}- (m -m') ^{2}} \right ] }{\sqrt{{m}^{2} - (m -m') ^{2}} } + O(\max\{{m}^2 j^{-2},  {m}^{-1} j^{-1}\}) 
\right \}
.
\end{align}
In the large-$j$ limit, we can further approximate this expression as an integral: making use of Eq.~\eqref{neq:wignerD.asymp.eff}, we have (henceforth still under the conditions $j \to\infty$, $m, m' = \omega(1)$, and $m = O( \sqrt{j})$)
\begin{align}
& 1- \sum _{m' <\sqrt{j} } |d ^{ j } _{m'm} (\beta _{m  } ) | ^{2} 
\nonumber \\
&~~= 1- \int ^{2} _{0} H\left (\frac{\sqrt{j}}{m} -x \right )
\frac{1+\sin \left[2 m (1-x) \arccos (1-x) - 2 m\sqrt{1- (1 -x) ^{2}} +\frac{\pi}{2}\right ] }{ \sqrt{1 - (1 -x) ^{2}} }
\frac{\dd{x}}{\pi }
+ O(\max\{{m}^2 j^{-\frac{3}{2}},  {m}^{-1} j^{-\frac{1}{2}}\})
\nonumber \\
\label{neq:prob.reset.int}
&~~= 1- \int ^{2} _{0} H\left (\frac{\sqrt{j}}{m} -x \right ) \frac{\dd{x}}{ \pi \sqrt{1 - (1 -x) ^{2}} }
+ O(\max\{{m}^2 j^{-\frac{3}{2}},  {m}^{-1} j^{-\frac{1}{2}}\})
. 
\end{align}
Here, $H(\cdot)$ denotes the Heaviside step function. Now we divide the analysis of Eq.~\eqref{neq:MC.crit.ratio.sep} into two cases. First, if $m<\sqrt{j}/2$, then the step function is always $1$, so the integral is $\int_0^2 \frac{\dd{x}}{ \pi \sqrt{1 - (1 -x) ^{2}}}=1$, and the reset probability in Eq.~\eqref{neq:prob.reset.int} is negligible. In this case, Eq.~\eqref{neq:MC.crit.ratio.sep} gives
\begin{align}  
m<\sqrt{j}/2: \quad 
& \sum _{m '} \mathcal{P} (m  \to m ' )  \frac{{m'}^{\alpha} }{m^{\alpha} } 
= \sum _{0\le m' \le 2 m } |d ^{ j } _{m' m} (\beta _{m  } ) | ^{2} \frac{  m'^{\alpha}  }{m^{\alpha} } 
+ O(\max\{{m}^2 j^{-\frac{3}{2}},  {m}^{-1} j^{-\frac{1}{2}}\})
.  
\end{align}
Substituting Eq.~\eqref{neq:wignerD.asymp.eff} into the above equation, we obtain
\begin{align} 
m<\sqrt{j}/2: \quad 
& \sum _{m' } \mathcal{P} (m  \to m ' )  \frac{m'^{\alpha} }{m^{\alpha} } 
\nonumber \\
&= \int ^{2} _{0}  x ^{\alpha}
\frac{1+\sin \left[2 m (1-x) \arccos (1-x) - 2 m\sqrt{1- (1 -x) ^{2}} +\frac{\pi}{2}\right ] }{ \sqrt{1 - (1 -x) ^{2}} }
\frac{\dd{x}}{\pi } 
+ O(\max\{{m}^2 j^{-\frac{3}{2}},  {m}^{-1} j^{-\frac{1}{2}}\})
\nonumber \\
\label{neq:ratio.bound.m.lower}
&= \int ^{2} _{0} 
\frac{ x ^{\alpha} \dd{x} }{\pi \sqrt{1 - (1 -x) ^{2}} } 
+ O(\max\{{m}^2 j^{-\frac{3}{2}},  {m}^{-1} j^{-\frac{1}{2}}\})
.
\end{align}
We can explicitly compute the integral on the right-hand side as 
\begin{align}
\int ^{2} _{0} 
\frac{ x ^{\alpha} \dd{x} }{\pi \sqrt{1 - (1 -x) ^{2}} } = \frac{1}{\pi }
\int ^{\pi} _{0} 
(1 -\cos \theta  )  ^{\alpha} \dd{\theta}
= \frac{2 ^{\alpha} }{\pi }
\int ^{\pi} _{0} 
\sin ^{2\alpha} \frac{\theta }{2} \dd{\theta}
= \frac{2^{\alpha}}{\pi} B (\alpha+\tfrac{1}{2},\tfrac{1}{2})
,
\label{eq:beta}
\end{align}
where $B(\cdot,\cdot)$ denotes the beta function.

Noting that $\sin ^{2\alpha} \frac{\theta }{2} < (\frac{\theta }{2}) ^{2\alpha} $ for all $\theta >0$, we can upper bound Eq.~\eqref{eq:beta} as 
\begin{align}
\frac{2 ^{\alpha} }{\pi }
\int ^{\pi} _{0} 
\sin ^{2\alpha} \frac{\theta }{2} \dd{\theta}
<\frac{2 ^{\alpha} }{\pi }
\int ^{\pi} _{0} \frac{\theta^{2\alpha}  }{2^{2\alpha} }\dd{\theta}
=\frac{\pi ^{2\alpha } }{2 ^{\alpha} (2\alpha+1) }
.
\label{eq:beta_upper_bound}
\end{align}
Expanding the function $\frac{ \pi  ^{2\alpha }}{ 2 ^{\alpha}}- (2\alpha +1) $ in a Taylor series, it is straightforward to show that this expression is strictly smaller than $0$ for $\alpha \in (0,0.1)$, so for such values of $\alpha$, $\pi^{2\alpha}/[2^\alpha(2\alpha+1)]<1$, and therefore
\begin{align} 
m = \omega (1), \, m<\sqrt{j}/2: \quad 
\sum _{m' } \mathcal{P} (m  \to m  ')  \frac{m'^{\alpha} }{m^{\alpha} } < 1
. 
\end{align}
(Alternatively, we show in the End Matter that Eq.~\eqref{eq:beta} is at most $1$ for all $\alpha \in (0,1)$.)

In the other regime $m>\sqrt{j}/2$, we compute the contribution from the reset separately. In this case, from Eq.~\eqref{neq:prob.reset.int}, we can write the reset probability as
\begin{align} 
\sqrt{j} \ge m>\sqrt{j}/2: \quad 
1- \sum _{m' <\sqrt{j} } |d ^{ j } _{m' m} (\beta _{m  } ) | ^{2} 
&=   1- \int ^{\frac{\sqrt{j}}{m} } _{0} \frac{\dd{x}}{ \pi \sqrt{1 - (1 -x) ^{2}} } + O(\max\{{m}^{3} j^{-2},  j^{-1}\}) 
\\
&= \frac{1}{2} - \frac{\arcsin(\frac{\sqrt{j}}{m} -1)}{\pi} + O({m}^{3} j^{-2} ) 
, 
\end{align}
so that, in this regime, we can derive an upper bound for Eq.~\eqref{neq:MC.crit.ratio.sep} as
\begin{align} 
\sum _{m' } \mathcal{P} (m \to m '  )  \frac{m'^{\alpha} }{m^{\alpha} }
 &< \left( \frac{1}{2} - \frac{\arcsin(\frac{\sqrt{j}}{m} -1)}{\pi} + O({m}^{3} j^{-2} ) \right ) \frac{(\sqrt{j}+1)^{\alpha} }{m^{\alpha}} 
+ \sum _{m' <\sqrt{j} } |d ^{ j } _{m'm} (\beta _{m  } ) | ^{2} \frac{m'^{\alpha}}{m^{\alpha} }
\nonumber \\
 &= \left( \frac{1}{2} - \frac{\arcsin(\frac{\sqrt{j}}{m} -1)}{\pi} + O({m}^{3} j^{-2} ) \right ) \left( \frac{\sqrt{j}}{m }  \right ) ^{\alpha} \left(1+ O(j^{-\frac{1}{2}})\right  )
\nonumber \\
&\quad+ \int ^{\frac{\sqrt{j}}{m}} _{0}  x ^{\alpha}
\frac{1+\sin \left[2 m (1-x) \arccos (1-x) - 2 m\sqrt{1- (1 -x) ^{2}} +\frac{\pi}{2}\right ] }{ \sqrt{1 - (1 -x) ^{2}} }
\frac{\dd{x}}{\pi } + O({m}^{3} j^{-2} ) 
\nonumber \\
&= \left( \frac{1}{2} - \frac{\arcsin(\frac{\sqrt{j}}{m} -1)}{\pi} \right ) \left( \frac{\sqrt{j}}{m }  \right ) ^{\alpha} 
+ \int ^{\frac{\sqrt{j}}{m}} _{0} 
\frac{ x ^{\alpha} \dd{x} }{\pi \sqrt{1 - (1 -x) ^{2}} } 
+ O({m}^{3-\alpha} j^{-2+\frac{\alpha}{2}} ) 
\nonumber \\
&= \left( \frac{1}{2} - \frac{\arcsin(\frac{\sqrt{j}}{m} -1)}{\pi} \right ) \left( \frac{\sqrt{j}}{m }  \right ) ^{\alpha} 
+ \frac{1}{\pi }
\int ^{\pi - \arccos(\frac{\sqrt{j}}{m}-1)} _{0} 
(1 -\cos \theta  )  ^{\alpha} \dd{\theta}
+ O({m}^{3-\alpha} j^{-2+\frac{\alpha}{2}} ) 
\nonumber \\
&= \frac{\arccos(\frac{\sqrt{j}}{m} -1)}{\pi} \left( \frac{\sqrt{j}}{m }  \right ) ^{\alpha} 
+ \frac{2 ^{\alpha}}{\pi }
\int ^{\pi - \arccos(\frac{\sqrt{j}}{m}-1)} _{0} 
\sin ^{2\alpha} \frac{\theta }{2} \dd{\theta}
+ O({m}^{3-\alpha} j^{-2+\frac{\alpha}{2}} ) 
\nonumber \\
\label{neq:ratio.approx.mp.l}
&< \frac{\arccos(\frac{\sqrt{j}}{m} -1)}{\pi} \left( \frac{\sqrt{j}}{m }  \right ) ^{\alpha} 
+ \frac{\left [\pi - \arccos(\frac{\sqrt{j}}{m}-1) \right ]^{2\alpha +1}}{\pi (2\alpha +1) 2 ^{\alpha}}
+ O({m}^{3-\alpha} j^{-2+\frac{\alpha}{2}} ) 
. 
\end{align}
We can show that the right-hand side of Eq.~\eqref{neq:ratio.approx.mp.l} is again asymptotically upper bounded by a number smaller than $1$. Specifically, setting $\arccos(\frac{\sqrt{j}}{m} -1) = \zeta$, we can rewrite the first two terms in Eq.~\eqref{neq:ratio.approx.mp.l} as
\begin{align}
\label{neq:ratio.m0.large}
\frac{\zeta}{\pi} \left( 1+\cos \zeta \right ) ^{\alpha} 
+ \frac{\left (\pi - \zeta \right ) ^{2\alpha +1}}{\pi (2\alpha +1) 2 ^{\alpha}}
, \quad 
\zeta \in \left(0,\frac{\pi}{2}\right)
.
\end{align}
For $\alpha \in (0,1)$, the function in Eq.~\eqref{neq:ratio.m0.large} monotonically decreases as $\zeta$ increases, so that 
\begin{align} 
\zeta \in (0,\frac{\pi}{2}): \quad 
\frac{\zeta}{\pi} \left( 1+\cos \zeta \right ) ^{\alpha} 
+ \frac{\left (\pi - \zeta \right ) ^{2\alpha +1}}{\pi (2\alpha +1) 2 ^{\alpha}}
\le \frac{ \pi  ^{2\alpha }}{ (2\alpha +1) 2 ^{\alpha}}
.
\end{align}
From our upper bound on Eq.~\eqref{eq:beta_upper_bound}, this is at most $1$ for $\alpha \in (0,0.1)$, so we obtain
\begin{align}
\sqrt{j} \ge m>\sqrt{j}/2: \quad 
& \frac{\arccos(\frac{\sqrt{j}}{m} -1)}{\pi} \left( \frac{\sqrt{j}}{m }  \right ) ^{\alpha} 
+ \frac{\left [\pi - \arccos(\frac{\sqrt{j}}{m}-1) \right ]^{2\alpha +1}}{\pi (2\alpha +1) 2 ^{\alpha}}
< 1 
\quad  \forall  \alpha \in (0,0.1)
.
\end{align}
Further, noting that $m \le \sqrt{j} $, we see that in the asymptotic $j\to \infty $ limit, the last term in Eq.~\eqref{neq:ratio.approx.mp.l}, $O({m}^{3-\alpha} j^{-2+\frac{\alpha}{2}} ) $, becomes $o(1)$, as 
\begin{align}
& {m}^{3-\alpha} j^{-2+\frac{\alpha}{2}} = \left ( \frac{m}{\sqrt{j}}\right ) ^{3-\alpha} j ^{- \frac{1}{2}}  \le j ^{- \frac{1}{2}}
.
\end{align}
Thus, we have shown that, in the asymptotic limit $j \to \infty$, the following inequality holds for every $m$ satisfying $m = \omega(1)$ and $m = O( \sqrt{j})$: 
\begin{align} 
\sum _{m ' } \mathcal{P} (m  \to m ' )  \frac{m'^{\alpha} }{m^{\alpha} } < 1, \quad  
\forall  \alpha \in (0,0.1) 
. 
\end{align}
Choosing any particular $\alpha \in (0,0.1)$ gives the result.
\end{proof}

\begin{proof}[Proof of Theorem \ref{thm:main}]
At time $0$, we have $m = j$ with probability $1$, so $\ev{[M(0)]^{\alpha}} = (\sqrt{j} + 1)^{\alpha}$, where we define $M (t) $ as a function of $m$ at time step $t$ as per Eq.~\eqref{neq:def.M.via.m}. By Lemma \ref{lem:ratio-small}, we have $\frac{\ev{[M (t+1)]^{\alpha}}}{\ev{[M  (t)]^{\alpha}}} < c$ for each $t$ with $M  (t) = \omega(1)$, so
\begin{equation}
\ev{[M(t)]^{\alpha}} < c^t \ev{[M(0)]^{\alpha}} = c^t (\sqrt{j} + 1)^{\alpha},
\end{equation}
unless there exists a $t'\le t$ such that $M(t') =O(1)$. 
Making use of the above inequality, for any desired $\eps > 0$, we can attain either $\ev{[M(t)]^{\alpha}} < \eps$ or $M(t) =O(1)$ in time
\begin{equation}
t _{\eps} = \frac{\alpha \log(\sqrt{j}+1) + \log(1/\eps)}{\log(1/c)} = O(\log j)
.
\end{equation}
In the former case, we have
\begin{equation}
\begin{aligned}
t \ge t _{\eps}: \quad 
\eps > \ev{[M(t)]^{\alpha} } &\ge \Pr[M(t)=0] \cdot 0 + \Pr[M(t) \ge 1] \cdot 1  = \Pr[M(t) \ge 1]  = 1 - \Pr[M(t) = 0],
\end{aligned}
\end{equation}
so that  $\Pr[M(t) = 0] > 1-\eps$. In the latter case, i.e., if we have $M(t) =O(1)$, from Eq.~\eqref{neq:wignerd.bessel} we conclude that any such state has a $\Theta(1)$ transition probability to reach the $m = 0$ state.
Thus in expectation repeating this $O(\log j)$-step procedure a constant number of times will yield the $m = 0$ state.
\end{proof}

\end{document}